\documentclass[sigconf, nonacm, pdfa]{acmart}

\usepackage[a-2b]{pdfx}
\usepackage{newfloat}
\usepackage{listings}
\usepackage{color}
\usepackage{mathrsfs}
\usepackage{graphicx}
\usepackage{multirow}
\usepackage{bm}
\usepackage{subfigure}
\usepackage{enumitem}
\usepackage{balance}
\usepackage{booktabs}
\usepackage[ruled]{algorithm2e}
\usepackage{marvosym}

\newcommand\vldbdoi{10.14778/3705829.3705859}
\newcommand\vldbpages{466 - 474}
\newcommand\vldbvolume{18}
\newcommand\vldbissue{2}
\newcommand\vldbyear{2024}
\newcommand\vldbauthors{\authors}
\newcommand\vldbtitle{\shorttitle} 
\newcommand\vldbavailabilityurl{URL_TO_YOUR_ARTIFACTS}
\newcommand\vldbpagestyle{empty} 

\begin{document}
\title{MILLION: A General Multi-Objective Framework with Controllable Risk for Portfolio Management}

\author{Liwei Deng}
\affiliation{%
  \institution{University of Electronic Science and Technology of China}
}
\email{deng_liwei@std.uestc.edu.cn}

\author{Tianfu Wang}
\affiliation{%
  \institution{University of Science and Technology of China}
}
\email{tianfuwang@mail.ustc.edu.cn}

\author{Yan Zhao}
\affiliation{%
  \institution{University of Electronic Science and Technology of China\Letter}
}
\email{yanz@cs.aau.dk}

\author{Kai Zheng}
\affiliation{%
  \institution{University of Electronic Science and Technology of China\Letter}
}
\email{zhengkai@uestc.edu.cn}
\authornote{Liwei Deng and Tianfu Wang are equally contributed to this work. \Letter Yan Zhao and Kai Zheng are corresponding authors. They are with School of Computer Science and Engineering, and Shenzhen Institute for Advanced Study, UESTC. }

\begin{abstract}
Portfolio management is an important yet challenging task in AI for FinTech, which aims to allocate investors' budgets among different assets to balance the risk and return of an investment. In this study, we propose a general \textbf{M}ulti-object\textbf{I}ve framework with contro\textbf{LL}able r\textbf{I}sk for p\textbf{O}rtfolio ma\textbf{N}agement (\textbf{MILLION}), which consists of two main phases, i.e., return-related maximization and risk control. Specifically, in the return-related maximization phase, we introduce two auxiliary objectives, i.e., return rate prediction, and return rate ranking, 
combined with portfolio optimization to remit the overfitting problem and 
improve the generalization of the trained model to future markets. Subsequently, in the risk control phase, we propose two methods, i.e., portfolio interpolation and portfolio improvement, to achieve fine-grained risk control and fast risk adaption to a user-specified risk level. For the portfolio interpolation method, we theoretically prove that the risk can be perfectly controlled if the to-be-set risk level is in a proper interval. In addition, we also show that the return rate of the adjusted portfolio after portfolio interpolation is no less than that of the min-variance optimization, as long as the model in the reward maximization phase is effective. Furthermore, the portfolio improvement method can achieve greater return rates while keeping the same risk level compared to portfolio interpolation. Extensive experiments are conducted on three real-world datasets.
The results demonstrate the effectiveness and efficiency of the proposed framework. 
\end{abstract}

\maketitle

\pagestyle{\vldbpagestyle}
\begingroup\small\noindent\raggedright\textbf{PVLDB Reference Format:}\\
\vldbauthors. \vldbtitle. PVLDB, \vldbvolume(\vldbissue): \vldbpages, \vldbyear.\\
\href{https://doi.org/\vldbdoi}{doi:\vldbdoi}
\endgroup
\begingroup
\renewcommand\thefootnote{}\footnote{\noindent
This work is licensed under the Creative Commons BY-NC-ND 4.0 International License. Visit \url{https://creativecommons.org/licenses/by-nc-nd/4.0/} to view a copy of this license. For any use beyond those covered by this license, obtain permission by emailing \href{mailto:info@vldb.org}{info@vldb.org}. Copyright is held by the owner/author(s). Publication rights licensed to the VLDB Endowment. \\
\raggedright Proceedings of the VLDB Endowment, Vol. \vldbvolume, No. \vldbissue\ %
ISSN 2150-8097. \\
\href{https://doi.org/\vldbdoi}{doi:\vldbdoi} \\
}\addtocounter{footnote}{-1}\endgroup

\ifdefempty{\vldbavailabilityurl}{}{
\vspace{.3cm}
\begingroup\small\noindent\raggedright\textbf{PVLDB Artifact Availability:}\\
The source code, data, and/or other artifacts have been made available at \url{https://github.com/LIWEIDENG0830/MILLION}.
\endgroup
}

\section{Introduction}
Portfolio management is an essential component of a trading system, which allocates a budget among different possible financial assets according to different objectives, such as maximizing returns at a given risk level~\cite{Markowitz1952PortfolioS,FinRLMeta,RobustPort,Wu2021PortfolioMS}. In 1952, Markowitz introduced a pioneer work, called Modern Portfolio Theory (MPT)~\cite{Markowitz1952PortfolioS}. MPT aims to construct a portfolio by solving a combinational optimization problem, leading to a higher return per risk than trading an individual asset~\cite{Zhang2020DeepLF}. Recently, the benefit of the portfolio compared with investing a single asset is further confirmed~\cite{ericzivot,Zhang2020DeepLF}, e.g., Eric Zivot~\cite{ericzivot} shows that the risk of a long-only portfolio is always lower than that of an individual asset, for a given expected return, as long as assets are not perfect correlated. Due to the desirable property of investing a bucket of assets, portfolio management has drawn much attention over the past decades.

Existing studies on portfolio management can be roughly divided into three main lines, i.e., predict-then-optimize, reinforcement learning (RL) based methods, and deep learning (DL) based approaches, based on different optimization ways. The methods along the first line~\cite{RobustPort,Geng2023RethinkingAB,Cheng2023AGF,Li2015MovingAR,Chen2021MeanvariancePO,Han2023EfficientCS,Zhu2022WISEWB} first estimate future price or return rate of each asset, and then solve a combinational optimization problem, e.g., mean-variance model, to obtain the final portfolio. For example, Huang et al.~\cite{RobustPort} forecast the future price of each asset through a sliding window-based moving average and then optimize the robust median reversion problem with the estimated prices to obtain the portfolio. Despite its ability of great ease of use, the performance of these methods is strongly related to the accuracy of the forecasting model. Unfortunately, the accurate asset price or return rate cannot be accessible due to the volatility of the dynamic market. Thus, rather than focusing on accurate price prediction, RL-based methods~\cite{Hao2023StockPM,Liu2021FinRLDR,FinRLMeta,Bai2023MercuryAD,Wang2019AlphaStockAB,Jang2023DeepRL} aim to directly obtain a portfolio from the observed market state through maximizing the defined reward function. 
For example, Liu et al.~\cite{Liu2021FinRLDR} propose a general RL-based framework to achieve automated trading, in which they adopt classical RL methods, such as Proximal Policy Optimization (PPO), to optimize the neural networks using the cumulative return as their reward. However, these methods ignore the fact that the portfolio optimization problem is different from others, such as video gaming and cheesing, where the reward from many investing goals, such as Sharpe ratio~\cite{Sharpe1994TheSR} and cumulative return, is differentiable. Therefore, optimizing from a surrogate loss in RL is inefficient. The methods in the last line~\cite{Zhang2021AUE,Zhang2020DeepLF,Xu2020RelationAwareTF,Ye2020ReinforcementLearningBP} overcome this issue through directly optimizing the objective. For example, Zhang et al.~\cite{Zhang2020DeepLF} develop an end-to-end DL-based model and directly adopt Sharpe ratio as their objective. 

Despite growing interest in the last direction, existing studies are still limited to the following aspects. 
First, the market is highly dynamic and the historical information of each asset has a relatively low signal-to-noise ratio, 
which causes a poor generalization of the trained model to the future market. Current studies on this problem either adopt simple architecture~\cite{Zhang2020DeepLF,Zhang2021AUE},
or incorporate outsourcing information~\cite{tfwang-kdd-2024-comet}, such as financial news~\cite{Ye2020ReinforcementLearningBP,Zhou2020DomainAM,Deng2019KnowledgeDrivenST,Li2021AME,Ma2022MultisourceAC,Ding2015DeepLF}, to enrich the information ratio in raw data. The former approaches may not be effective enough to extract the necessary features,
while the outsourcing information may not be always available in the latter approaches. 
Second, different investors may have different tastes in taking risks. For example, an investor with limited wealth usually prefers to allocate their budget to low-risk bonds and stocks, while an investor with enough spare money prefers to take more risks to achieve higher expected returns. Current learning-based approaches only return a fixed portfolio for a given market state, which cannot satisfy different investors' demands with different risk levels. Moreover, existing DL-based methods cannot achieve fine-grained risk control. For example, Zhang et al.~\cite{Zhang2021AUE} combine the return objective and risk term using a Lagrange form, where the risk multiplier is hard to preset, i.e., higher weights risk underfitting the return while lower weights weaken risk control.



To tackle these problems, we propose a general \textbf{M}ulti-object\textbf{I}ve framework with contro\textbf{LL}able r\textbf{I}sk for p\textbf{O}rtfolio ma\textbf{N}agement, named MILLION. In summary, we decompose the mean-variance optimization~\cite{Markowitz1952PortfolioS} into two phases, i.e., return-related maximization and risk control. Specifically, in the first phase, rather than introducing outsourcing data to improve the data quantity, we only leverage the assets' historical price and volume information despite our framework also being easy to modify to incorporate other source information. To improve the generalization of the trained model to the future market, we introduce two auxiliary objectives, i.e., assets return rate forecasting and assets return rate ranking, which are highly related to portfolio construction. 
Furthermore, to control risk to a user-specified risk level, we propose a simple but effective portfolio interpolation method in the risk control phase, in which the constructed portfolio is obtained through interpolation between the portfolios from the reward maximization phase and the min-variance optimization. We theoretically prove that the risk can be perfectly controlled as long as the given risk level is in a risk interval. Besides, we also theoretically show that the expected return of the portfolio after interpolation is always greater than the portfolio from the min-variance optimization if the reward maximization is effective. In addition, based on the idea of the interpolation approach, we further propose a portfolio improvement method to achieve higher portfolio return with the same risk level compared with the interpolation method. It should be emphasized that the risk control components can be appended to any existing portfolio construction approaches to help them control the risk. 

In summary, our contributions are as follows:

\begin{itemize}[leftmargin=*]
\item We propose a general DL-based framework with multi-objective learning to improve the generalization of the trained model to perform better in the future market.
\item We develop two risk control approaches, i.e., portfolio interpolation and portfolio improvement, which can be used to fit investors' personalized risk preferences. According to the best of our knowledge, this is the first attempt to achieve fine-grained risk control in learning based portfolio construction. 
\item We conduct extensive experiments to demonstrate the effectiveness of MILLION and its components. 
\end{itemize}

\section{Preliminaries}
We present necessary concepts and define the problem addressed.

\begin{figure*}[t]
\centering
\includegraphics[width=0.86\linewidth]{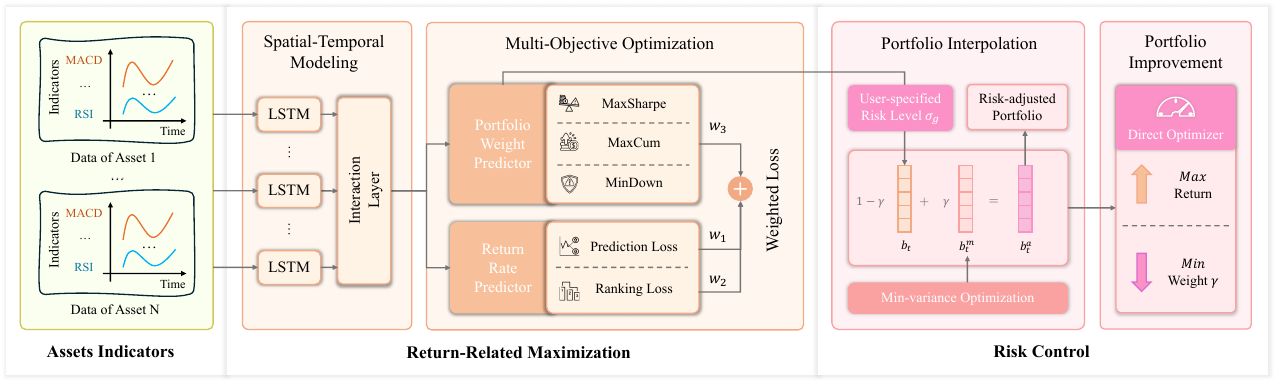}
\vspace{-0.2cm}
\caption{The proposed MILLION framework}
\label{fig:framework}
\end{figure*}


\begin{definition}[Holding Period]
\emph{
A holding period is the minimum time unit to invest an asset. We follow previous studies that divide the whole investment period into multiple non-overlapped holding periods with fixed length, such as one day or one month. 
}
\end{definition}

\begin{definition}[Asset Prices]
\emph{
The price of an asset is defined as a time series $\bm{p}^{i}=\{p_1^{i}, p_2^{i}, ..., p_t^{i}\}$, where $p_t^{i}$ is the price of asset $i$ at $t$.
}
\end{definition}

\begin{definition}[Return Rate]
\emph{
The return rate of an asset $i$ at time $t$ is defined as $r_t^i=p_{t+1}^{i}/p_t^{i}-1$, which presents that if an investor spent $n$ cash to buy asset $i$ at time $t$, he can get profit $n * r_t^i$. 
}
\end{definition}

\begin{definition}[Risk]
\emph{
Following the risk definition in MPT~\cite{Markowitz1952PortfolioS}, we define the risk as volatility (i.e., variance) of return rate (i.e., $\sigma_i^2=E[(r_t^i-\overline{r}^i)^2]$, where $\overline{r}^i$ is the expectation of $r^i$). The idea behind this definition is that the return rate of an asset has a lower variance, and the certainty of investing in this asset is higher, which induces lower risk. In this study, $\sigma_i^2$ is calculated from a sliding window of historical asset return rate, in which the window size is consistent with $w$ in temporal modeling (cf. Section~\ref{sec:spatiotemporal}). 
}
\end{definition}

\begin{definition}[Long Position]
\emph{
The long position is to buy an asset $i$ at time $t_1$ and then sell it at time $t_2$, aiming for a profit derived from an increase in the asset's price over this period. The produced profit can be formulated as $n * r^i_{t_1, t_2}$, where $n$ and $r^i_{t_1, t_2}$ denotes the investment amount  and return rate from $t_1$ to $t_2$ of asset $i$, respectively.
}
\end{definition}

Conversely, a short position operation represents the opposite strategy, where investors profit from the decline in an asset's price.
In this study, we prohibit the short position operation, i.e., investors can only get profits when the prices of invested assets increase. 


\begin{definition}[Portfolio]
\emph{
Given $N$ assets to be invested, a portfolio is defined as a vector $\bm{b}=(b^{1}, b^{2}, ..., b^{N})$, where $b^{i}$ presents the proportion of the investment on asset $i$ and $\sum_{i=1}^{N} b^{i}=1$. The return rate of portfolio $\bm{b}$ is $\bm{b}^T\bm{r}$, where $\bm{r}$ indicates the return rates of N assets. The risk of portfolio is $\bm{b}^T\bm{\Sigma}\bm{b}$, where $\bm{\Sigma} \in \mathcal{R}^{N\times N}$ is the return rate covariance matrix of N assets. 
}
\end{definition}

\begin{definition}[Portfolio Management]
\emph{
Portfolio management is a sequential investment, which determines portfolio $\bm{b}$ at the end of each holding period. We pursue to achieve two goals in this study: (1) maximizing the portfolio return; (2) controlling the portfolio risk to a user-specific risk level. 
}
\end{definition}

\section{Methodology}

We propose a general multi-objective framework with controllable risk for portfolio management, named MILLION, as shown in Figure~\ref{fig:framework}, which consists of two main phases, i.e., return-related maximization and risk control. In the return-related maximization phase, we adopt a DL-based spatio-temporal model to construct a portfolio for each time step to maximize the objectives (e.g., Sharpe ratio) that are related to the portfolio return, which provides a significant signal to get profits from the investment behaviors. Then, in the risk control phase, we propose two novel methods (i.e., portfolio interpolation and portfolio improvement) to adjust the portfolio obtained from the return-related maximization phase to fulfill a user-specified risk level. Each component is elaborated in the following sections.

\subsection{Return Maximization}
Earning money is the primary goal for most investors, in which understanding the current market state is a fundamental and critical step to provide instructions to guide the investors' decision process. In this study, rather than designing a new DL-based model to effectively extract more powerful and predictive representations from raw features, we focus on developing a general framework that can fit various models. Thus, we directly adopt existing models, i.e., LSTMHA~\cite{Feng2018EnhancingSM}, with slight modifications to encode the market state. 
It should be noted that another model with the ability of spatial-temporal modeling can also be adopted to encode the state of assets~\cite{deng2024learning,tfwang-tsc-2023-hrl-acra,liuitransformer,wu2021autoformer,zhou2021informer}.
As shown in the left panel of Figure~\ref{fig:framework}, we model the spatial (i.e., relations among assets) and temporal (i.e., relations along timestamp) information with attention technique and LSTM, respectively, in which the covariance matrix between assets' historical return rate is incorporated into the attention module. After the representation of each asset is obtained, a multi-objective optimization module is appended, in which two portfolio-related objectives are adopted to improve the model generalization.

\subsubsection{Assets Indicators} 
\label{sec:indicators}
Following previous studies, e.g., FinRL~\cite{Liu2021FinRLDR}, we incorporate the eight indicators derived from the prices and volumes as our model input. These indicators include Moving Average Convergence/Divergence (MACD), Bollinger Bands (BOLL) (i.e., lower bound and upper bound of BOLL), Relative Strength Index (RSI), Commodity Channel Index (CCI), Directional Movement Index (DMI), and Simple Moving Average (SMA) (i.e., 30-days and 60-days).
Given the disparate scales across these financial metrics, we employ Z-score normalization to standardize them.

\subsubsection{Spatio-Temporal Modeling} 
\label{sec:spatiotemporal}
To construct an effective portfolio, we necessitate insight into the future market, particularly regarding the performance (e.g., return rate) of individual assets. To achieve this goal, we differentiate the modeling of each asset into temporal and spatial relations. 

\noindent \textbf{Temporal Modeling.} We use the vector $\bm{x}_t^i$ to denote the history state of asset $i$ at time $t$, which consists of eight indicators as stated in section~\ref{sec:indicators}. Thus, the current state of asset $i$ at time $t$ is presented by a window size of $\bm{x}_t^i$ (i.e., $\bm{X}_t^i=\{\bm{x}_{t-w}^i, \cdots, \bm{x}_t^i\}$). A one-layer $\mathit{LSTM}$ is adopted to recursively encode $\bm{X}^i$ into a vector. 
\begin{equation}
\footnotesize
\bm{h}^i_t = \mathit{LSTM}(\bm{X}_t^i)
\end{equation}

\noindent where $\bm{h}^i_t \in \mathcal{R}^{d}$ is the representation of asset $i$ at time $t$, which models the intra-correlation along the timestamps. 

\noindent \textbf{Spatial Modeling.} Despite the temporal dependency is encoded into $\bm{h}^i_t$, the inter-correlation among assets are not included. 
Therefore, we leverage the attention technique to model the dynamic relations among assets, in which the covariance matrix between assets’ historical return rates is integrated to remit the burden in the learning process. 
\begin{equation}
\footnotesize
\begin{aligned}
\hat{\bm{h}}_t^i = \frac{1}{\beta+1} \sum_{k=1}^{N} \alpha_k \cdot \bm{h}_t^{k} + \frac{\beta}{\beta+1} \sum_{k=1}^{N} c_{i,k} \cdot \bm{h}_t^{k} \\
{\color{black}{\alpha_k = \frac{\exp(\bm{W}\bm{h}_t^k)}{\sum_{j=1}^{N} \exp(\bm{W}\bm{h}_t^j)} \ \ \ \ \ \ \ \ \ \ \ \ \ \ \ \ \ \ }}
\end{aligned}
\end{equation}
\noindent where $\bm{W}_a\in \mathcal{R}^{d_1\times d}$ is the to-be-learned parameters, $\beta$ is a scalar to balance the weight between attention weights and covariance matrix, which is updated with the training goes on, $c_{i,k}$ presents the covariance between asset $i$ and $k$, and $\hat{\bm{h}}_t^i$ is the encoded representation for asset $i$ at time $t$. 

\noindent \textbf{Portfolio Construction.} Based on the extracted feature from previous modules, we adopt a two-layer $\mathit{MLP}$ with $\mathit{ReLU}$ activation to evaluate each asset and construct a portfolio as follows:
\begin{equation}
\footnotesize
\begin{aligned}
\label{eq:portfolio_construction}
v^i_t &= {\color{black}{MLP_p(\bm{\hat{h}}_t^i)}} \\
b^i_t &= \frac{exp(v^i_t)}{\sum_{k=1}^{N} exp(v^k_t)}
\end{aligned}
\end{equation}

\noindent where $v^i_t$ indicates the estimated valuation of investing on asset $i$, and $b^i_t$ is the portfolio weight of asset $i$.

\subsubsection{Multi-Objective Optimization} 
After obtaining the portfolio weight at each time step, we can directly optimize the Sharpe ratio or the cumulative return as the same with previous studies~\cite{Zhang2020DeepLF,Zhang2021AUE}. 
However, while neural networks offer powerful representational ability, they often face challenges in generalizing to future market conditions effectively. 
To remit this problem, we propose a shift from a singular objective to a multiple objective optimization framework, which incorporates two auxiliary loss functions: asset return rate prediction and asset return rate ranking. In the latter, we first elaborate the portfolio optimization and then the auxiliary optimization, separately. 

\noindent \textbf{Portfolio Optimization.} To optimize the constructed portfolio, there are two common objectives~\cite{Zhang2020DeepLF} as follows:
\begin{equation}
\label{eq:maxsharpe}
\footnotesize
\begin{aligned}
(MaxCum) \mathcal{L}_{mc} &= \prod_{t=1}^{T} (r_t^p + 1) \\
(MaxSharpe) \mathcal{L}_{ms} &= \frac{Mean(\{r_t^p\}_{t=1}^{T})}{Std(\{r_t^p\}_{t=1}^{T})} \\
r_t^p &= {\color{black}{\bm{b}_t^T \bm{r}_t - c_t |\bm{b}_t^T - \bm{b}_{t-1}^T|_1}}
\end{aligned}
\end{equation}

\noindent where $\bm{r}_t$ represents the return rate of each assets at time $t$, $r_t^p$ is the corresponding portfolio return rate, $c_t$ is the transaction cost rate, and $T$ is the total number of holding period. $\mathit{MaxCum}$ focuses on maximizing the cumulative return, which will drive the model to construct a centralized portfolio (i.e., investing in a single asset that may achieve the highest return rate among other assets). $\mathit{MaxSharpe}$ not only focuses on maximizing the portfolio return rate at each time but also aims to minimize its standard deviation, which will impose the model to construct a relatively conservative portfolio. Except for the two commonly used objectives, we also develop another objective as follows:
\begin{equation}
\label{eq:mindown}
\footnotesize
(MinDown) \mathcal{L}_{md} = - \sum_{t=1}^{T}\max(-r_t^p + \delta_d, 0)
\end{equation}
\noindent where $\delta_d$ indicates a threshold, which presents the investors' expected return rate in each holding period. The goal of $\mathit{MinDown}$ is to construct portfolios that can achieve a return rate larger than the given threshold $\delta_d$. The lower $\delta_d$ is, the more conservative the constructed portfolio is, which will endow the model with roughly risk-control ability. Moreover, $\delta_d$ is unnecessary to be fixed, which can be dynamic according to the situation of the current market state. For example, 
a simple implementation is to replace $\delta_d$ with the return rate of a benchmark such as NAS100 index.

\noindent \textbf{Auxiliary Optimization.} Directly optimizing a single objective in portfolio optimization, the trained model may suffer from the overfitting problem (i.e., it has poor generalization in the future market) due to the highly dynamic market and low signal-to-noise ratio in historical information. To remit this problem, we introduce two auxiliary objectives~\cite{Ding2015DeepLF,Gao2021GraphBasedSR,Zheng2023RelationalTG,Feng2018TemporalRR,Wang2022AdaptiveLP}, which are optimized combining with the objective in portfolio optimization. 
\begin{equation}
\footnotesize
\begin{aligned}
\label{eq:prediction_obj}
(Prediction) \mathcal{L}_{p} &= {\color{black}{\sum_{t=1}^{T} ||\hat{\bm{r}}_{t} - \bm{r}_{t}||_2}} \\
(Ranking) \mathcal{L}_{r} &= \sum_{t=1}^{T} \sum_{i=1}^{N} \sum_{j=1}^{N} max(-(\hat{r}_t^{i}-\hat{r}_{t}^{j})(r_t^i - r_t^j), 0)
\end{aligned}
\end{equation}

\noindent where $\hat{\bm{r}_t}=[\hat{r}_t^1, \hat{r}_t^2, \cdots, \hat{r}_t^N]$ and $\bm{r}_t=[r_t^1, r_t^2, \cdots, r_t^N]$ is the predicted and the ground-truth return rate of N assets. $\hat{\bm{r}_t}$ is obtained through a $\mathit{MLP}$ with $\hat{\bm{h}}_t^i$ as inputs. This neural networks share the same architecture with $\mathit{MLP}_p$ in portfolio construction but with different parameters. 
With these objectives, we define our final optimization objective as follows:
\begin{equation}
\footnotesize
\mathcal{L} = -\zeta_m \mathcal{L}_{*} + \zeta_p \mathcal{L}_p + \zeta_r \mathcal{L}_r
\end{equation}
\noindent where $*$ is in $\{mc, ms, md\}$, and $\zeta_m$, $\zeta_p$, and $\zeta_r$ are the weights to balance the contributions of different objectives. In our experiments, we notice the performance of the constructed portfolio is sensitive to the weights. To reduce the efforts to search an optimal weights, we follow Kendal et al.~\cite{Kendall2017MultitaskLU} to set the weights adaptively. 
\begin{equation}
\label{eq:return_max}
\footnotesize
\mathcal{L} = -\frac{1}{\zeta_m^2} \mathcal{L}_{*} + \frac{1}{\zeta_p^2} \mathcal{L}_p + \frac{1}{\zeta_r^2} \mathcal{L}_r + \sum_{i\in \{m,p,r\}} \log \zeta_{i}
\end{equation}

\noindent where $\zeta_m$, $\zeta_p$, and $\zeta_r$ are to-be-learned parameters, which are adaptively updated in the training phase. The parameters of the neural networks are optimized through minimizing $\mathcal{L}$ in Equation~\ref{eq:return_max}. 

\subsection{Risk Control}

Despite risk management being a vital component in portfolio management, current DL-based and RL-based approaches are hard to achieve fine-grained risk control. 
For example, 
Zhang et al.~\cite{Zhang2021AUE} utilize neural networks to optimize the Lagrange formulation mean-variance model, in which the risk term is weighted and added to the portfolio return. It can achieve rough risk control in the training phase but without guaranteeing the unseen data. Moreover, existing studies always construct one portfolio for all investors at each holding period, which cannot satisfy investors' different preferences for risk-taking. Thus, we propose two methods, i.e., portfolio interpolation and portfolio improvement, to deal with these issues, which are elaborated as follows.

\subsubsection{Portfolio Interpolation} We denote $\bm{b}_t^m$ as the portfolio at time $t$ obtained from the min-variance optimization, which has the lowest risk compared to any other portfolios. To control the risk of the constructed portfolio from the return maximization phase, we obtain the risk-adjusted portfolio $\bm{b}_t^a$ with interpolation as follows:
\begin{equation}
\footnotesize
\label{eq:interpolation}
\bm{b}_t^a = (1-\gamma_t) \bm{b}_t +  \gamma_t \bm{b}_t^m
\end{equation}

\noindent where $\gamma_t \in [0,1]$ is the weight to control the amount of interpolation. 
With the above interpolation method, we have the following proposition, in which we denote $\sigma^a_t$, $\sigma_t$, and $\sigma_t^m$ as the risk of the three portfolios $\bm{b}_t^a$, $\bm{b}_t$, and $\bm{b}_t^m$, respectively (e.g., $\sigma_t=\bm{b}_t^T\Sigma_t\bm{b}_t$).
\begin{proposition} 
\label{prop:risk}
$\sigma^a_t$ is a decreasing monotone function in terms of $\gamma_t$
if $\delta_t \ne \delta_t^m$, whose value is in the interval [$\sigma_t^m$, $\sigma_t$].
\end{proposition}
\begin{proof}
\vspace{-0.3cm}
\begin{equation}
\footnotesize
\begin{aligned}
\label{eq:risk_control}
\sigma^a_t &= {\bm{b}_t^a}^T \Sigma_t \bm{b}_t^a \\
&=[(1-\gamma_t) \bm{b}_t +  \gamma_t \bm{b}_t^m]^T \Sigma_t [(1-\gamma_t) \bm{b}_t +  \gamma_t \bm{b}_t^m] \\
&=(\bm{b}_t^T\Sigma_t\bm{b}_t - 2\bm{b}_t^T\Sigma_t\bm{b}_t^m + {\bm{b}^m_t}^T\Sigma_t\bm{b}^m_t) \gamma_t^2 
\\
&\ \ \ \ \ + 2(\bm{b}_t^T\Sigma_t\bm{b}_t^m - \bm{b}_t^T\Sigma_t\bm{b}_t) \gamma_t + \bm{b}_t^T\Sigma_t\bm{b}_t
\end{aligned}
\end{equation}
We represent the portfolios with the eigenvectors of $\Sigma_t$ (i.e., $\bm{b}_t=\sum_{i=1}^{N}c_i\bm{x}_i$ and $\bm{b}_t^m=\sum_{i=1}^{N}d_i\bm{x}_i$). The eigenvalue of $\bm{x}_i$ is denoted as $\lambda_i$. It should be noted that $\Sigma_t$ is symmetric and semi-positive whose eigenvalue is non-negative (i.e., $\lambda_i\ge0, \forall i\in[1,N]$). 
\begin{equation}
\footnotesize
\begin{aligned}
\bm{b}_t^T\Sigma_t\bm{b}_t - 2\bm{b}_t^T\Sigma_t\bm{b}_t^m + {\bm{b}^m_t}^T\Sigma_t\bm{b}^m_t &= \sum_{i=1}^{N} \lambda_i c_i^2 - 2\sum_{i=1}^{N}\lambda_i c_i d_i + \sum_{i=1}^{N} \lambda_i d_i^2 \\
&= \sum_{i=1}^{N} \lambda_i (c_i - d_i)^2 \ge 0
\end{aligned}
\end{equation}
When $\sigma_t^a$ is a quadratic function in terms of $\gamma_t$, we can know that this function is an upward opening and it gets the lowest when $\gamma_t$ equals $1$. Thus, this function is decreasing monotone when $\gamma_t$ varies from $0$ to $1$ (i.e., $\sigma_t^a \in [\sigma_t^m, \sigma_t]$). Next, when the quadratic term equals $0$, this function degenerates to a linear function. It should be noted that $\sigma_t^m \le {\bm{b}^*_t}^T\Sigma_t\bm{b}^*_t$ from the definition of $\bm{b}_t^m$, where $\bm{b}^*_t$ represents any portfolio. We analyze the coefficient of the primary term as follows:
\begin{equation}
\footnotesize
\begin{aligned}
2(\bm{b}_t^T\Sigma_t\bm{b}_t^m - \bm{b}_t^T\Sigma_t\bm{b}_t) &= 2\bm{b}_t^T\Sigma_t\bm{b}_t^m - \bm{b}_t^T\Sigma_t\bm{b}_t - {\bm{b}_t^m}^T\Sigma_t\bm{b}_t^m - \\
& \ \ \ \ \ \ \bm{b}_t^T\Sigma_t\bm{b}_t + {\bm{b}_t^m}^T\Sigma_t\bm{b}_t^m \\
&= - \sum_{i=1}^{N} \lambda_i (c_i - d_i)^2 - (\sigma_t - \sigma_t^m) \le 0
\end{aligned}
\end{equation}
We can see that when the function is linear, it is also decreasing monotone when $\gamma_t$ in $[0,1]$. Thus, combining these two situations, we can conclude that $\sigma_t^a \in [\sigma_t^m, \sigma_t]$ and with the increase of $\gamma_t$, the portfolio risk will be gradually decreased if $\sigma_t^m \ne \sigma_t$. 
\vspace{-0.2cm}
\end{proof}

With proposition~\ref{prop:risk}, we can control the portfolio risk to a user-specified risk level $\sigma_g$ by replacing the left of Equation~(\ref{eq:risk_control}) with $\sigma_g$ and solving $\gamma_t$ as long as $\sigma_g \in [\sigma_t^m, \sigma_t]$. Since this equation is a quadratic equation with only $\gamma_t$ unknown, it is easy to calculate the exact value of $\gamma_t$ when the expected risk is given. Therefore, the portfolio can be fast adapted to satisfy different investors' requests in terms of risk level.
After analysing the effect of risk through interpolation, we continue to study the effect on the portfolio return rate. We have the following proposition, in which we denote $r_t^a$, $r^p_t$, and $r_t^m$ as the return rate of the three portfolios $\bm{b}_t^a$, $\bm{b}_t$, and $\bm{b}_t^m$, respectively (e.g., $r^p_t=\bm{b}_t^T\bm{r}_t$). 
\begin{proposition}
\label{prop:return}
$r_t^a$ is always no less than $r_t^m$ if the model
is effective (i.e., $r_t^p \ge r_t^m$). 
\vspace{-0.3cm}
\end{proposition}
\begin{proof}
\begin{equation}
\footnotesize
\begin{aligned}
r_t^a &= [(1-\gamma_t)\bm{b}_t + \gamma_t\bm{b}_t^m]^T \bm{r}_t
= (1-\gamma_t)\bm{b}_t^T \bm{r}_t + \gamma_t{\bm{b}_t^m}^T\bm{r}_t \\
&= (1-\gamma_t)r_t^p+\gamma_t r_t^m
\end{aligned}
\end{equation}
From above equation, we can see that the portfolio return rate is a linear function in terms of $\gamma_t$, which achieves the lowest $r_t^m$ when $\gamma_t$ equals to $1$. Thus, $r_t^a$ is no less than $r_t^m$. 
\vspace{-0.2cm}
\end{proof}

Proposition~\ref{prop:return} demonstrates that no matter how we interpolate, the return rate of $\bm{b}_t^a$ is always bounded by it of $\bm{b}_t$ (i.e., upper bound) and $\bm{b}_t^m$ (i.e., lower bound), which means that the portfolio interpolation method is safe and will not generate a portfolio that causes dramatic loss. 

\begin{figure}[t]
\centering
\subfigure[Interpolation]{\includegraphics[width=0.34\linewidth]{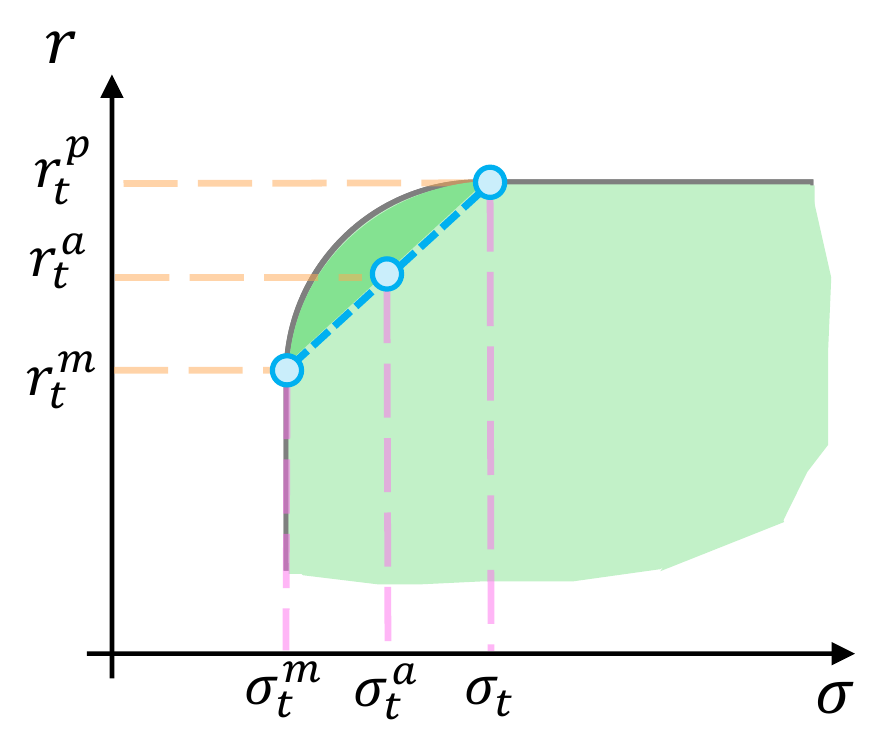}}
\subfigure[Improvement]{\includegraphics[width=0.34\linewidth]{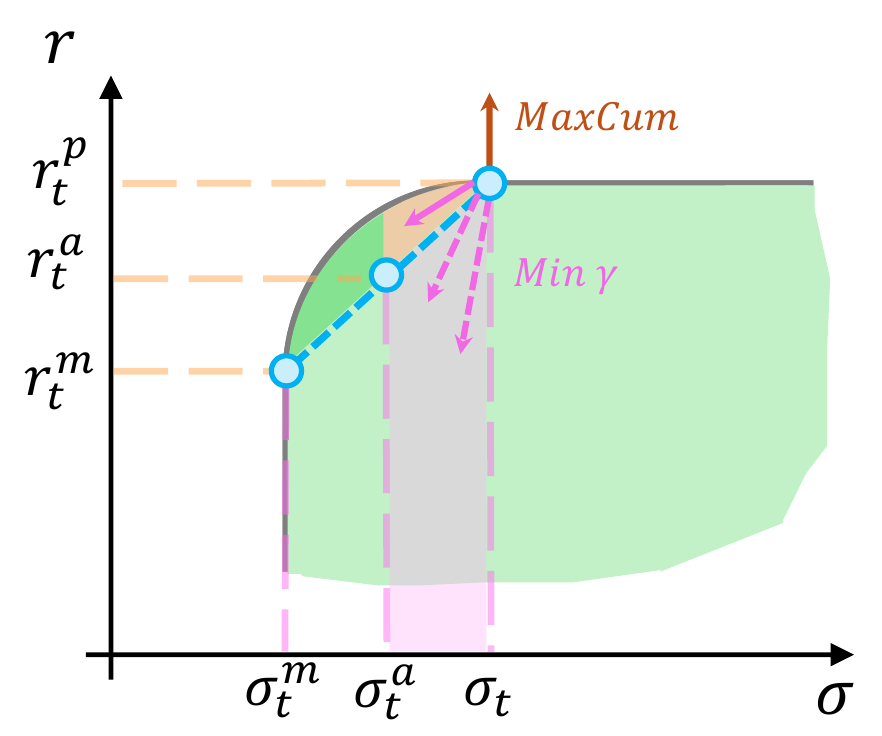}}
\vspace{-0.4cm}
\caption{An illustration of portfolio interpolation.}
\label{fig:illustration_interpolation}
\vspace{-0.55cm}
\end{figure}


\subsubsection{Portfolio Improvement}
Despite the advantages of the portfolio interpolation method, the portfolio after interpolation may be far away from the efficient frontier, which means there is an opportunity to improve the portfolio return while keeping the user-specified risk unchanged. We provide an illustration of this situation as shown in Figure~\ref{fig:illustration_interpolation}(a). The horizontal and vertical axis represent the risk and return, respectively. The blue dash line presents the portfolio interpolation
. From this figure, we can see that there may exist an orange area that cannot be obtained through interpolation, in which there exist points that have the same risk as the interpolated point but have higher returns. To reach these points, we propose a portfolio improvement approach, which is to optimize the portfolio from return maximization to push it to the orange area as shown in Figure~\ref{fig:illustration_interpolation}(b). Then, the portfolio interpolation is adopted to control risk, in which the interpolated point is expected to have a higher return with the same risk compared with the portfolio from the pure portfolio interpolation method. Before introducing our approach, we first present a proposition as follows:

\begin{proposition}
Assume we have a set of portfolios $\{\bm{b}_t^1, \bm{b}_t^2, \cdots, \bm{b}_t^i, $ $\cdots\}$ whose return rates are the same. We apply the portfolio interpolation method to control their risks to a given risk $\sigma_g$. The portfolio with the highest return after interpolation has the lowest $\gamma_t$.
\end{proposition}
\begin{proof}
Since this proposition can be directly deduced from proposition~\ref{prop:return}, the details of the proof are omitted. 
\end{proof}

\begin{table*}[]
\footnotesize
\caption{Performance comparison between MILLION and competitors on three real-world datasets.}
\vspace{-0.3cm}
\label{tab:return_maximization}
\begin{tabular}{c|ccccc|ccccc|ccccc}
\hline
\multicolumn{1}{c|}{Methods} & \multicolumn{5}{c|}{DOW30}                  & \multicolumn{5}{c|}{NAS100}                 & \multicolumn{5}{c}{Crypto10}                  \\ \cline{2-16} 
\multicolumn{1}{c|}{}        & APR$\uparrow$     & AVOL$\downarrow$     & ASR$\uparrow$     & ACR$\uparrow$ & MDD$\uparrow$     & APR$\uparrow$     & AVOL$\downarrow$     & ASR$\uparrow$     & ACR$\uparrow$ & MDD$\uparrow$     & APR$\uparrow$      & AVOL$\downarrow$     & ASR$\uparrow$     & ACR$\uparrow$  & MDD$\uparrow$     \\ \hline
Market & 0.2182 & {\color{black}{0.1133}} & 1.7994 & 3.6065 & -0.0605 & 0.3566 & 0.3704 & 1.0101 & 1.1099 & -0.3213 & -0.0494 & 0.5612 & 0.1922 & -0.0656 & -0.7530 \\ 
MVM & 0.0941 & 0.1138 & 0.8472 & 1.6677 & \textbf{-0.0565} & 0.3907 & \textbf{0.2914} & 1.2796 & 1.3885 & -0.2814 & 0.1406 & \textbf{0.4560} & 0.5180 & 0.1939 & -0.7252 \\ 
DT & 0.0948 & 0.1229 & 0.7979 & 1.1003 & -0.0861 & 0.9011 & 0.4100 & \underline{1.7731} & \underline{4.4232} & \textbf{-0.2037} & -0.1786 & 0.6937 & 0.0627 & -0.2116 & -0.8438 \\ 
LR & 0.2125 & 0.1416 & 1.4318 & 2.5832 & -0.0823 & 0.7128 & 0.4149 & 1.5057 & 2.2923 & -0.3109 & -0.1489 & 0.6744 & 0.0757 & -0.1774 & -0.8395 \\ 
RF & 0.2326 & 0.1562 & 1.4173 & 2.6315 & -0.0884 & 0.5249 & 0.3772 & 1.3082 & 1.7548 & -0.2991 & -0.1698 & 0.7046 & 0.0830 & -0.2175 & -0.7806 \\ 
SVM & 0.0047 & \textbf{0.1105} & 0.0980 & 0.0545 & -0.0871 & 0.8349 & \underline{0.3617} & 1.6598 & 3.1439 & \underline{-0.2656} & 0.0025 & 0.7448 & 0.3526 & 0.0028 & -0.8955 \\ 
LSTM-PTO & 0.1808 & 0.1194 & 1.4516 & 2.6715 & -0.0677 & -0.4360 & 0.9906 & -0.0831 & -0.7148 & -0.6099 & 0.0334 & 0.6555 & 0.3759 & 0.0432 & -0.7730 \\ 
LSTMHAM-PTO & 0.1044 & 0.1143 & 0.9254 & 1.8236 & -0.0572 & -0.4640 & 0.8107 & -0.3547 & -0.7971 & -0.5821 & -0.0297 & 0.7656 & 0.3403 & -0.0436 & \underline{-0.6799} \\ 
FinRL-A2C & 0.2186 & 0.1134 & 1.8001 & 3.5902 & -0.0609 & 0.3630 & 0.3718 & 1.0203 & 1.1285 & -0.3217 & -0.0451 & 0.5605 & 0.2000 & -0.0602 & -0.7488 \\ 
FinRL-PPO & 0.2379 & {\color{black}{\underline{0.1128}}} & 1.9492 & 4.1736 & \underline{-0.0570} & 0.3572 & 0.3686 & 1.0143 & 1.1170 & -0.3198 & -0.0835 & \underline{0.5521} & 0.1204 & -0.1095 & -0.7624 \\ 
LSTMHAM-S                    & 0.2731 & 0.1212 & 2.0532 & 3.8952 & -0.0701 & 0.4442 & 0.5038 & 0.9794 & 1.5705 & -0.2828 & 0.1837  & 0.6858 & 0.5684 & 0.2525  & -0.7278 \\ 
LSTMHAM-C                    & 0.2915 & 0.1198 & 2.1949 & 4.1515 & -0.0702 & 1.2153 & 0.8541 & 1.3296 & 2.7640 & -0.4396 & 0.0838  & 0.6444 & 0.4357 & 0.1138  & -0.7365 \\ 
LSTMHAM-M                    & 0.3214 & 0.1279 & 2.2432 & 4.2906 & -0.0749 & 0.1846 & 0.4949 & 0.5888 & 0.5421 & -0.3404 & 0.0532  & 0.5797 & 0.3790 & 0.0714  & -0.7449 \\ \hline
\textbf{MILLION-S}                    & \textbf{0.3936} & 0.1417 & \underline{2.4132} & \underline{4.6153} & -0.0806 & \textbf{1.9763} & 0.7243 & \textbf{1.8528} & \textbf{6.9817} & -0.2830 & \textbf{0.3871}  & 0.7405 & \textbf{0.7914} & \textbf{0.5821}  & \textbf{-0.6650} \\
\textbf{MILLION-C}                    & \underline{0.3857} & 0.2103 & 1.6572 & 3.8870 & -0.0992 & \underline{1.3172} & 0.8578 & 1.3759 & 3.1930 & -0.4125  & \underline{0.1996}  & 0.7279 & \underline{0.5888} & \underline{0.2786}  & -0.7165 \\ 
\textbf{MILLION-M}                    & 0.3389 & 0.1288 & \textbf{2.5306}  & \textbf{4.6412} & -0.0736 & 1.2206 & 0.5638 & 1.6917 & 4.0832 & -0.2989 & 0.0581  & 0.5797 & 0.3870 & 0.0783  & -0.7440 \\ \hline
\end{tabular}
\vspace{-0.2cm}
\end{table*}

From this proposition, the portfolio improvement approach is to minimize the interpolation weight $\gamma_t$ for a user-specified risk level $\sigma_g$. Specifically, we solve Equation~\ref{eq:risk_control} for a given $\sigma_g$ to get $\gamma_t$, in which $\gamma_t$ can be presented as a function of $\sigma_g$, $\Sigma_t$, $\bm{b}_t$, and $\bm{b}_t^m$. Since $\gamma_t$ is the root of a quadratic equation with one unknown, the function $f$ is differentiable. Thus, we can directly minimize $\gamma_t$ to optimize $\bm{v}_t$ (i.e., the portfolio before $\mathit{Softmax}$, cf. Equation~\ref{eq:portfolio_construction}) to compel the portfolio constraint (i.e., {\color{black}{$\sum_{i=1}^{N} b^i = 1$}}). 
\begin{equation}
\label{eq:imp}
\footnotesize
\mathcal{L}_{imp} = \sum_{t=1}^{T} \gamma_t
\end{equation}

Moreover, since the goal of personalized risk control is to fit different investors' requests, it is time-consuming to finetune the whole model for each investor. Therefore, the model's parameters are fixed at the portfolio improvement, which can save the amount of calculation. 

However, only minimizing $\gamma_t$ may cause unexpected damage to the portfolio return as shown in Figure~\ref{fig:illustration_interpolation}(b), where the pink area will produce a smaller $\gamma_t$ but also a smaller portfolio return. To deal with this situation, we incorporate a return objective when optimizing. The return-added objective can be formulated as follows:
\begin{equation}
\label{eq:imp_return}
\footnotesize
\mathcal{L}_{imp+ret} = -\zeta \prod_{t=1}^{T} (\bm{b}_{t}^T \hat{\bm{r}_t}) + \sum_{t=1}^{N} \gamma_t
\end{equation}

\noindent where $\hat{\bm{r}_t}$ is the predicted assets' return rate, and $\zeta$ is a weight to balance the loss of the two components. It should be noted that the predicted return rate is not constrained to obtained from our predictor in the return-related maximization phase but can be accessed from any predictors, such as LSTM or more advanced models, and training separately. 

\section{Experiments}



\noindent \textbf{Dataset.} 
The data from both the U.S. stock market and cryptocurrency market is obtained using FinRL\footnote{https://github.com/AI4Finance-Foundation/FinRL} and CCXT\footnote{https://github.com/ccxt/ccxt}, and then preprocessed with the FinRL library to extract the indicators detailed in section~\ref{sec:indicators}. In the case of the stock market, we specifically focus on stocks from two prominent U.S. stock indexes, namely NAS100 and DOW30. For the cryptocurrency market, we select the top 10 cryptocurrencies by market occupancy.
Table~\ref{tab:statistic} provides an overview of the statistics for each dataset. In the case of NAS100 and DOW30 datasets, the data from the last year of the training period is utilized as a validation set for conducting model selection.



\noindent \textbf{Competitors.} We compare three types of competitors. (1) The predict-then-optimize methods are: Decision Tree (DT)~\cite{Breiman2017PointsOS}, Linear Regression (LR), Random Forest (RF)~\cite{Breiman2001RandomF}, Support Vector Machine (SVM)~\cite{Platt1999ProbabilisticOF}, LSTM-PTO~\cite{Hochreiter1997LongSM}, and LSTMHAM-PTO, {\color{black}{where they focus on predicting the assets' return rate of next holding period and then solve the classical mean-variance problem through maximizing Sharpe ratio. It should be noted that LSTMHAM-PTO shares the same model architecture with ours.}} (2) The RL-based methods: A2C and PPO, in which we adopt a RL-based financial lib, i.e., FinRL~\cite{FinRLMeta,Liu2021FinRLDR}, to implement. (3) {\color{black}{The DL-based methods: LSTMHAM-S, LSTMHAM-C, and LSTMHAM-M, which keep the same model architecture as ours but with only one objective to optimize, where -S, -C, and -M indicates the model is trained with $MaxSharpe$, $MaxCum$, and $MinDown$, respectively.}} We also include two classical methods, i.e., Market and min-variance model (MVM), as competitors. 

\noindent \textbf{Evaluation Metrics.} We include six commonly-used metrics~\cite{Ye2020ReinforcementLearningBP,Zhang2020DeepLF,Cheng2023AGF} in portfolio management to evaluate the proposed framework, i.e., Cumulative Wealth (CW), Annualized Percentage Rate (APR), Annualized Volatility (AVOL), Annualized Sharpe Ratio (ASR), Maximum DrawDown (MDD), and Annualized Calmar Ratio (ACR). 

\noindent \textbf{Parameter Settings.} 
The holding period is fixed at one day, with a window size ($w$) of 20 for temporal modeling. Transaction cost ($c_t$) in Equation~\ref{eq:maxsharpe} is set to 0, while the threshold ($\delta_d$) in Equation~\ref{eq:mindown} is set to 0.005. 
Optimization is conducted using the $\mathit{AdamW}$ optimizer with a learning rate of 1e-4. The neural networks' hidden size ($d$) is set to 64 for the DOW30 and Crypto10 datasets and 128 for the NAS100 dataset. In subsequent sections, we independently assess the efficacy of our proposed components: return-related maximization and risk control.


\subsection{Return-Related Maximization Performance}

\noindent \textbf{Overall.} 
Table~\ref{tab:return_maximization} displays the backtesting outcomes of various models across three datasets, while Figure~\ref{fig:all_cumulative} illustrates cumulative wealth curves. It's evident that among the predict-then-optimize approaches, no single model outperforms others consistently in terms of APR or ASR. This suggests that different market conditions favor different prediction models, making it challenging to devise a universal model applicable to all scenarios. Additionally, the performance of these methods varies significantly across datasets, indicating the sensitivity of portfolio construction to predicted asset returns. Furthermore, Market and MVM exhibit relatively lower AVOL and MDD yet competitive APR compared to predict-then-optimize methods, underscoring the validity of diversified investment strategies. Notably, RL-based methods excel in DOW30 but not in other datasets. Our experiments reveal that RL-based methods often demand extensive interactions to learn effective policies on training data but struggle with generalization to unseen test data, highlighting training efficiency and generalizability issues. Moreover, LSTMHAM with diverse objectives generally outperforms other benchmarks in most cases. {\color{black}{Ultimately, our proposed MILLION framework consistently achieves the best performance in terms of return-related metrics such as APR, ASR, and ACR.}}

\noindent \textbf{Effect of different frameworks.} Comparing the same model architecture with different portfolio construction methods (e.g., LSTMHAM-PTO, FinRL-PPO, LSTMHAM-S), we can find that the DL-based framework always achieves better while the predict-then-optimize framework is comparable with RL-based framework. 

\noindent \textbf{Effect of different objectives.} 
Our investigation into the impact of different optimization objectives uncovers distinct advantages associated with each. {\color{black}{For instance, the $MaxSharpe$ objective typically minimizes MDD to a greater extent compared to other objectives. On the other hand, $MaxCum$ tends to achieve superior APR, while $MinDown$ consistently minimizes AVOL across most scenarios.}}

\noindent \textbf{Effect of multiple objectives.} Compared with single objective DL-based methods (e.g., LSTMHAM-S), our proposed framework MILLION (e.g., MILLION-S) can always perform better in terms of return-related metrics, which shows the effectiveness of using multiple objectives to train the DL-based models. 

\begin{table}[t]
\footnotesize   
\caption{Dataset Statistics}
\vspace{-0.3cm}
\label{tab:statistic}
\begin{tabular}{cccc}
\hline
\textbf{Dataset}  & \textbf{Training Period}          & \textbf{Testing Period}           & \textbf{\#Assets} \\ \hline
NAS100   & 2009/01/01-2020/02/01 & 2020/02/01-2021/01/01 & 78     \\ \hline
DOW30    & 2009/01/01-2019/02/01 & 2019/02/01-2020/01/01 & 28     \\ \hline
Crypto10 & 2018/08/02-2021/07/01 & 2021/07/01-2023/10/32 & 10     \\ \hline
\end{tabular}
\vspace{-0.55cm}
\end{table}

\begin{figure*}[t]
\centering
\begin{minipage}[l]{0.64\linewidth}
\centering
\includegraphics[width=\linewidth]{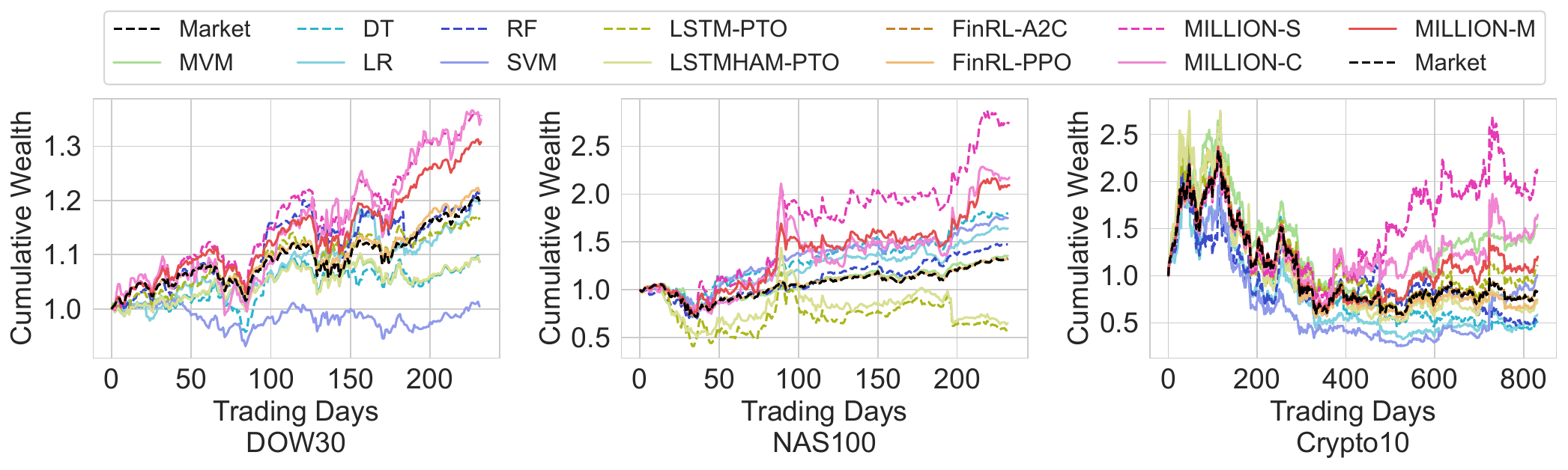}
\caption{{\color{black}{Backtest results in terms of CW on three datasets}}}
\label{fig:all_cumulative}
\end{minipage}
\begin{minipage}[c]{0.3\linewidth}
\centering
\includegraphics[width=0.89\linewidth]{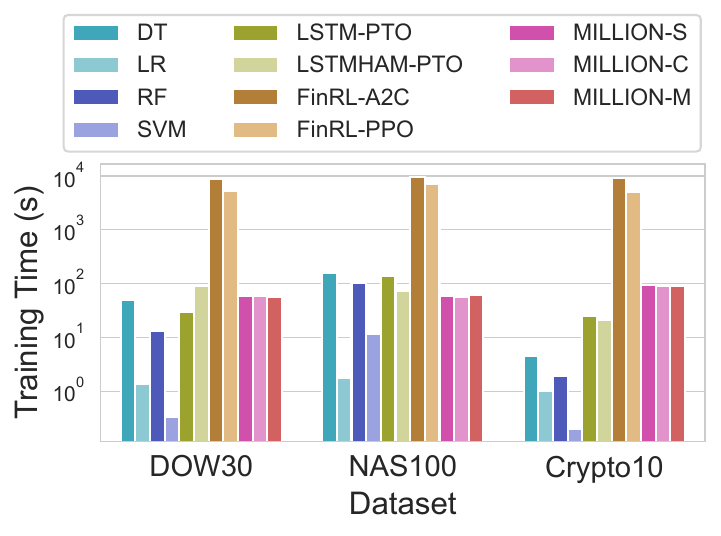}
\caption{{\color{black}{Training efficiency}}}
\label{fig:efficiency}
\end{minipage}
\end{figure*}

\subsection{Risk Control Performance}
We conduct a comparative analysis of the proposed risk control methodologies alongside predict-then-optimize strategies, as illustrated in Figure~\ref{fig:risk_control}. 
{\color{black}{The horizontal axis represents user-defined risk levels, while the vertical axis denotes the corresponding backtesting metrics. "Interpolation" and "Improvement" refer to portfolio interpolation (cf. Equation~\ref{eq:interpolation}) and portfolio improvement (cf. Equation~\ref{eq:imp}) methodologies, respectively. "Improvement-LSTM" and "Improvement-SVM" denote portfolio enhancement utilizing Equation~\ref{eq:imp_return}, where return rates are forecasted using LSTM and SVM algorithms, respectively.}}


From Figure~\ref{fig:risk_control}, it is evident that predict-then-optimize methodologies yield varied outcomes when risk is constrained to specific values across different datasets. Notably, for LR, ASR in the DOW30 dataset demonstrates a gradual increase with risk levels ranging from 1e-5 to 1e-4, while in NAS100, it exhibits an inverse trend. Furthermore, the construction of risk-constrained portfolios may outperform those optimized solely for maximizing the Sharpe ratio, as evidenced by Table~\ref{tab:return_maximization}. For example, in the DOW30 dataset, constraining LR's risk to 5e-5 results in an ASR exceeding 1.5, surpassing the 1.43 ASR achieved through maximizing the Sharpe ratio alone (cf. Table~\ref{tab:return_maximization}). Additionally, our proposed risk control methodologies consistently demonstrate an anticipated pattern: as specified risk increases, both realized ASR and AVOL simultaneously increase. Moreover, it is noteworthy that Improvement consistently outperforms Interpolation, even in scenarios where there are no predicted return rates to guide optimization. Incorporating predicted return rates into portfolio enhancement tends to yield superior performance by guiding the portfolio towards regions with higher return rates.


\subsection{Efficiency}
In Figure~\ref{fig:efficiency}, we present the training time required for each algorithm to reach convergence. Traditional machine learning methods, such as DT and RF, are executed on the CPU, while other models are run on an A100 GPU. Notably, MILLION demonstrates comparable training times to predict-then-optimize algorithms and outperforms RL-based methods in terms of speed. This is attributed to the stability of training in DL compared to RL, which typically requires over 100K steps to converge. Additionally, the interaction between RL agents and the environment is slower, further hampering the training efficiency of RL-based methods.

\begin{figure}[t]
\centering
\includegraphics[width=0.9\linewidth]{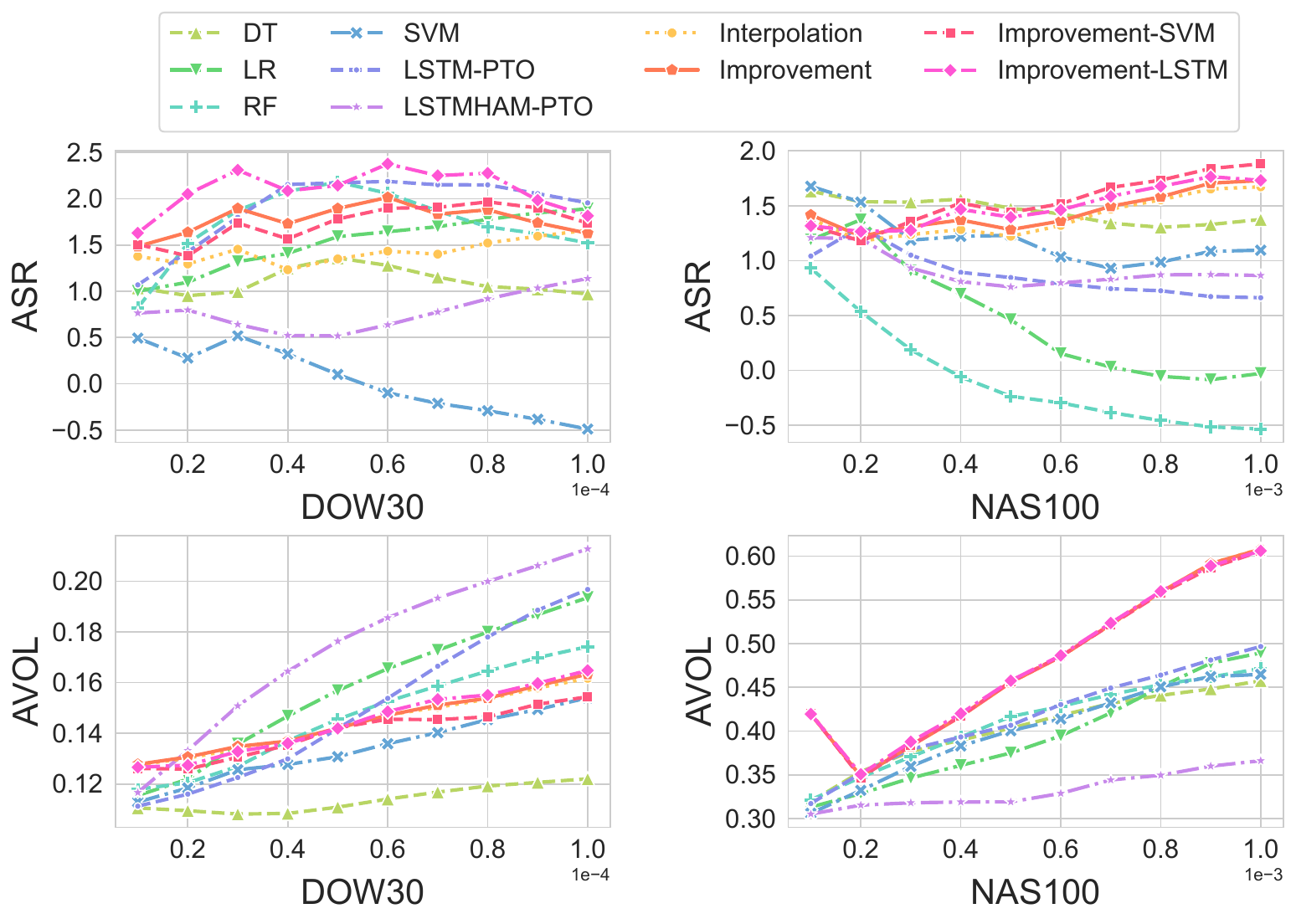}
\vspace{-0.3cm}
\caption{Comparison of the proposed risk control methods with predict-then-optimize approaches}
\label{fig:}
\label{fig:risk_control}
\vspace{-0.55cm}
\end{figure}

\subsection{Case Study}
\label{sec:case_study}
To further understand the effectiveness of the proposed risk control approaches, we conduct a case study on DOW30 dataset. 

\noindent \textbf{Effect of uniform interpolation.} 
In Figure~\ref{fig:case_study}(a), the performance of uniform portfolio interpolation is depicted, showcasing the variation of $\gamma_t$ from 0 to 1 with a step size of 0.1. The bottom curve represents the cumulative wealth attained from min-variance optimization, where $\gamma_t$ is fixed at 1 for all $t$. Notably, as $\gamma_t$ increases, both CW and AVOL gradually decrease. Furthermore, CW consistently remains bounded by the cumulative wealth obtained from min-variance optimization (MVM), underscoring the empirical validity of Proposition~\ref{prop:return}.


\noindent \textbf{Effect of portfolio improvement.} 
Apart from the portfolio interpolation method, we demonstrate the effectiveness of portfolio improvement in Figure~\ref{fig:case_study}(b). Here, we maintain the portfolio risk at 5e-5 and iterate Equation~(\ref{eq:imp}) optimization 30 times. Each line in the figure represents a single iteration, with the green line depicting the outcome of portfolio interpolation and the red line representing the final result of portfolio improvement. Notably, all lines exhibit identical risk levels. From this visualization, we observe a notable enhancement in the CW value, rising from approximately 1.2 to 1.35 by the conclusion of the test period, underscoring the effectiveness of the proposed portfolio improvement. Additionally, we examine the performance of portfolio improvement with varying numbers of optimization epochs (refer to Equation~\ref{eq:imp} and Equation~\ref{eq:imp_return}), showcased in Figure~\ref{fig:portfolio_improvement}, where each data point corresponds to the same risk level. Incorporating the predicted return rate into the portfolio improvement yields a final portfolio with superior returns and reduced risk, despite potential inaccuracies in the predicted return rate. For instance, the predict-then-optimize framework yields a 0.1808 APR, as demonstrated in Table~\ref{tab:return_maximization}. This observation underscores the robustness of our proposed portfolio improvement to the accuracy of predicted return rates, owing to the provision of a solid initial portfolio through the return-related maximization model.

\begin{figure}[t]
\centering
\subfigure[Interpolation]{
\includegraphics[width=0.47\linewidth]{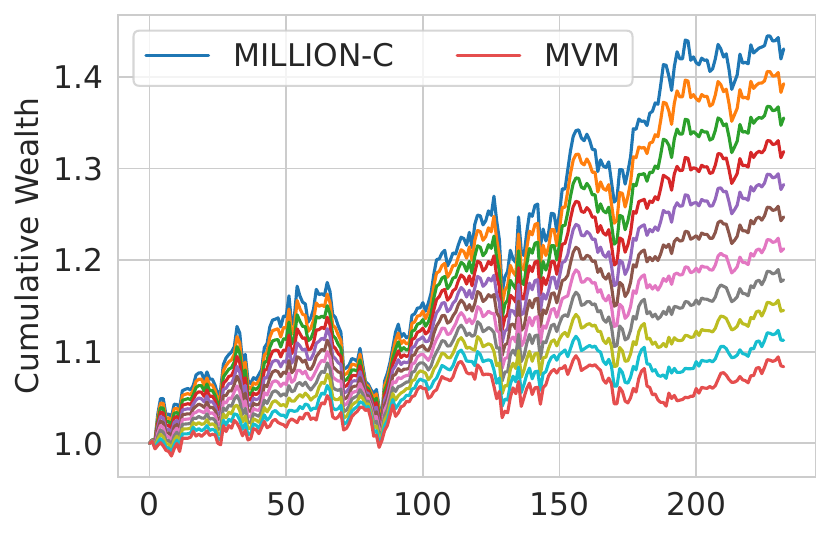}
}
\subfigure[Improvement]{
\includegraphics[width=0.45\linewidth]{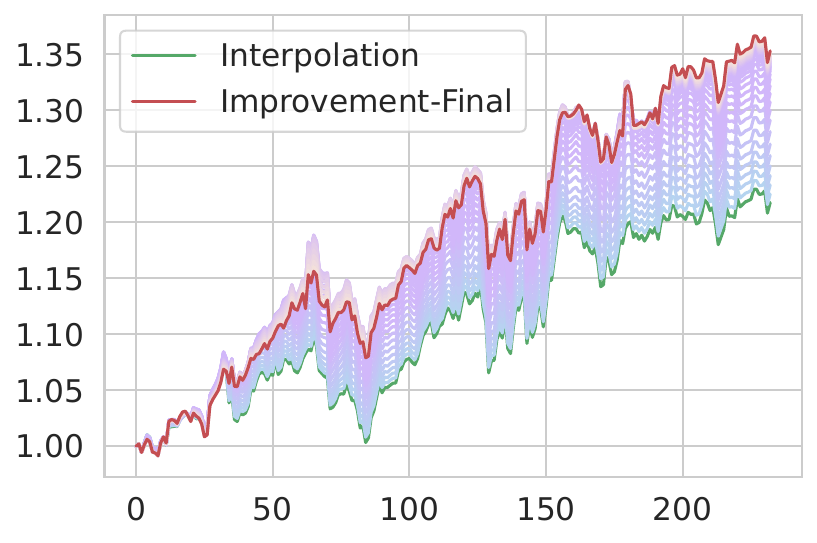}
}
\vspace{-0.1cm}
\caption{{\color{black}{The effect of the proposed risk control approaches in terms of CW on DOW30 dataset}}}
\label{fig:case_study}
\end{figure}

\begin{figure}[t]
\centering
\includegraphics[width=0.86\linewidth]{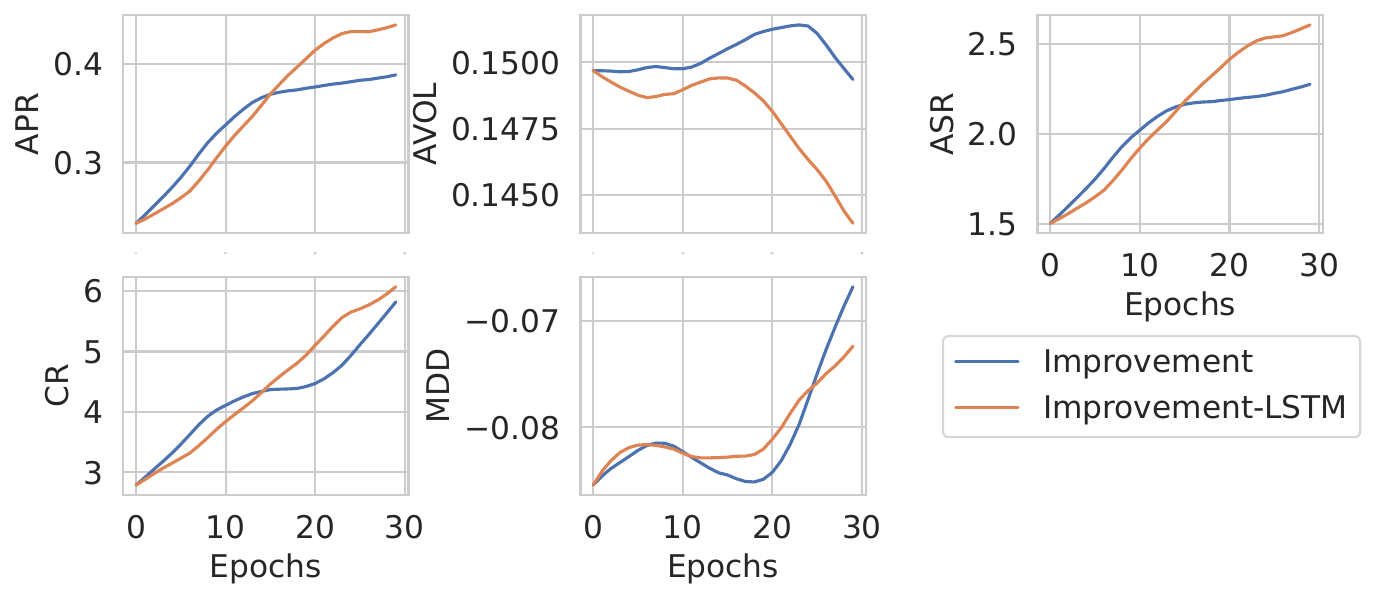}
\vspace{-0.1cm}
\caption{{\color{black}{The effect of portfolio improvement on metrics}}}
\label{fig:portfolio_improvement}
\end{figure}



\section{Related Works}
In this section, we survey the related studies on predict-then-optimize and direct portfolio optimization.

\noindent \textbf{Predict-then-optimize Portfolio Optimization.} Mean-variance model~\cite{Markowitz1952PortfolioS} is a classical method to construct a portfolio through solving a combinational optimization problem as follows:
\begin{equation}
\footnotesize
\begin{aligned}
\color{black}{
\max \sum_{i=1}^{N} b_t^i r_t^i \quad \text{s.t.} \quad \bm{b}_t^T\Sigma\bm{b}_t \le \sigma_g , \quad \sum_{i=1}^{N} b^i=1, \quad 1\ge b^i\ge 0, \forall_i
}
\end{aligned}
\end{equation}

\noindent in which the return rate of each asset, i.e., $\bm{r}_t$, and its covariance, i.e., $\bm{\Sigma}_t$, are supposed to be already known or simply estimated using the sample mean and sample covariance of historical assets' return rate. However, since the market is dynamic and volatile, the simple estimation may not reflect the future market. To construct a more effective portfolio, plenty of researchers are devoted to developing more powerful return rate prediction models. For example, Li et al.~\cite{Li2011ConfidenceWM} propose Confidence Weighted Mean Reversion (CWMR) to estimate the next price relative as the inverse of the last of it. Huang et al.~\cite{RobustPort} exploit the reversion phenomenon via robust $L_1$-median estimators to predict the next price relative. With the fast development of machine learning, a number of advanced models are proposed. For example, 
Li et al.~\cite{Li2021AME} propose a multimodal event-driven LSTM model using online news to predict stocks' return rates. Besides improving the accuracy of return rate prediction, there are lots of efforts in formulating a different portfolio optimization problem. For example, 
Rockafellar et al.~\cite{Rockafellar2000OptimizationOC} adopt conditional value-at-risk (CVaR) as the risk metric for portfolio construction. Furthermore, Lai et al.~\cite{Lai2022MultitrendCV} propose a multi-trend CVaR as the risk metric to optimize the constructed portfolio. Despite the methods in this line that can easily control risk to a user-specified risk level, i.e., $\sigma_g$, their effectiveness will be strongly influenced by the estimation of assets' return rate. Unfortunately, the accurate return rate prediction is difficult or even impossible, which may be due to the incorrectness of the model construction or the less predictability of the market states.

\noindent \textbf{Direct Portfolio Optimization.} Instead of predicting the return rate of assets, RL-based~\cite{Lien2023ContrastiveLA,Sun2021DeepScalperAR,Yang2022AST,Xu2020RelationAwareTF,Wang2021DeepTraderAD,Wang2020CommissionFI,Liu2020AdaptiveQT} and DL-based approaches aim to directly output a portfolio for each holding period, which usually can achieve higher investment profits compared 
with methods in the predict-then-optimize framework. In RL-based methods, they focus on learning a policy to map the market state to an action to maximize the discounted cumulative reward, in which the action is defined as the portfolio weight $\bm{b}_t$. For example, Liu et al.~\cite{FinRLMeta,Liu2021FinRLDR} develop a general deep RL-based framework to enable investors to automate trading, in which investors can flexibly incorporate their prior knowledge. 
Lien et al.~\cite{Lien2023ContrastiveLA} adopt contrastive learning technique and reward smoothing to indicate the stock relationships and maximize a long-term profit, respectively. 
For the DL-based methods~\cite{Zhang2020DeepLF,Zhang2021AUE}, the only difference to RL-based methods is the way of optimization, where RL-based methods optimize the parameters through gradients backward from a surrogate loss while DL-based methods directly optimize the portfolio objective through gradient ascent since this objective is differentiable. 
Zhang et al.~\cite{Zhang2020DeepLF} propose to directly optimize Sharpe ratio in different model architectures such as MLP, CNN, and LSTM. 
Although optimizing Sharpe ratio can control risk to some extent, they are hard to achieve fine-grained risk control and impossible to fit investors' personality in terms of risk preference with only one constructed portfolio. Moreover, the strong ability of representation of neural networks will degenerate the generalization of the trained model to be applied to the future market. 

\section{Conclusion}
In this paper, we propose a general multi-objective framework with controllable risk for portfolio management called MILLION, in which we decompose the portfolio management into two main phases, i.e., return-related maximization and risk control. In the first phase, we follow the DL-based portfolio framework and demonstrate that the multi-objective design is useful in improving return-related metrics, such as APR, and ASR, through backtesting on three real-world datasets. In the risk control phase, we propose two approaches to adjust the risk of the constructed portfolio to fit different investors' preferences in terms of risk-taking. Compared with methods in the predict-then-optimize framework, MILLION performs better in terms of ASR under the same risk level. 

\begin{acks}
This work is partially supported by NSFC (No. 62472068),  Shenzhen Municipal Science and Technology R\&D Funding Basic Research Program (JCYJ20210324133607021), and Municipal Government of Quzhou under Grant (No. 2023D044), and Key Laboratory of Data Intelligence and Cognitive Computing, Longhua District, Shenzhen.
\end{acks}

\balance
\bibliographystyle{ACM-Reference-Format}
\bibliography{sample}

\end{document}